\documentclass[conference]{IEEEtran}
\usepackage{cite}
\ifCLASSINFOpdf
\usepackage[pdftex]{graphicx}
\graphicspath{./figures}
\DeclareGraphicsExtensions{.pdf,.jpeg,.png}
\else
\fi
\usepackage[cmex10]{amsmath}
\usepackage{array}
\usepackage{url}
\usepackage[utf8]{inputenc}  
\usepackage{citesort}
\usepackage[T1]{fontenc} 
\usepackage{mathtools}
\usepackage{amsfonts}
\usepackage{amssymb}
\usepackage{graphicx}
\usepackage{bm}
\usepackage{amsthm}
\usepackage{balance}
\def\(({\left(}
\def\)){\right)}                       
\def\[[{\left[}
\def\]]{\right]}
\usepackage{color}

\usepackage{xspace,prettyref}

\newcommand{\iiddistr}{{\stackrel{\text{\iid}}{\sim}}}

\newcommand{\reals}{{\mathbb{R}}}

\newcommand{\eexp}{{\rm e}}

\newcommand{\diff}{{\rm d}}

\newcommand{\expect}[1]{\mathbb{E}\left[ #1 \right]}

\newcommand{\prob}[1]{ \mathbb{P}\left\{ #1 \right\} }

\newcommand{\var}{\mathsf{var}}

\newcommand{\iid}{i.i.d.\xspace}
\newrefformat{eq}{(\ref{#1})}
\newrefformat{chap}{Chapter~\ref{#1}}
\newrefformat{sec}{Section~\ref{#1}}
\newrefformat{alg}{Algorithm~\ref{#1}}
\newrefformat{fig}{Fig.~\ref{#1}}
\newrefformat{tab}{Table~\ref{#1}}
\newrefformat{rmk}{Remark~\ref{#1}}
\newrefformat{clm}{Claim~\ref{#1}}
\newrefformat{def}{Definition~\ref{#1}}
\newrefformat{cor}{Corollary~\ref{#1}}
\newrefformat{lmm}{Lemma~\ref{#1}}
\newrefformat{prop}{Proposition~\ref{#1}}
\newrefformat{app}{Appendix~\ref{#1}}
\newrefformat{hyp}{Hypothesis~\ref{#1}}
\newrefformat{thm}{Theorem~\ref{#1}}
\newrefformat{ass}{Assumption~\ref{#1}}

\newcommand{\Iprod}[2]{\langle #1, #2 \rangle}

\newcommand{\calA}{{\mathcal{A}}}
\newcommand{\calB}{{\mathcal{B}}}

\newcommand{\calE}{{\mathcal{E}}}

\newcommand{\calG}{{\mathcal{G}}}
\newcommand{\calH}{{\mathcal{H}}}

\newcommand{\calJ}{{\mathcal{J}}}

\newcommand{\calN}{{\mathcal{N}}}

\newcommand{\calS}{{\mathcal{S}}}

\newcommand{\be}{\begin{equation}}
\newcommand{\ee}{\end{equation}}
\newcommand{\bea}{\begin{eqnarray}}
\newcommand{\eea}{\end{eqnarray}}

\DeclareMathAlphabet{\varmathbb}{U}{bbold}{m}{n}

\newcommand{\EE}{\mathbb{E}}

\newcommand{\vct}[1]{\bm{#1}}
\newcommand{\mtx}[1]{\bm{#1}}
\renewcommand{\hat}{\widehat}
\renewcommand{\tilde}{\widetilde}

\newtheorem{theorem}{Theorem}[section]
\newtheorem{lemma}[theorem]{\textbf{Lemma}}
\newtheorem{thm}[theorem]{\textbf{Theorem}}

\begin{document}
\title{
Mutual Information in 
Rank-One Matrix Estimation}
\author{\IEEEauthorblockN{ Florent
    Krzakala\IEEEauthorrefmark{1}\IEEEauthorrefmark{2}\IEEEauthorrefmark{4}, Jiaming Xu\IEEEauthorrefmark{4}
 and
    Lenka Zdeborov\'a\IEEEauthorrefmark{3}\IEEEauthorrefmark{4}} 
 \IEEEauthorblockA{\IEEEauthorrefmark{1} Laboratoire de Physique
    Statistique, Sorbonnes Universit\'es \&
    Universit\'e Pierre \& Marie Curie, 75005    Paris}
  \IEEEauthorblockA{\IEEEauthorrefmark{2} Laboratoire de Physique
    Statistique (CNRS UMR-8550), PSL Universit\'es \& \'Ecole Normale Sup\'erieure, 75005
    Paris}
 \IEEEauthorblockA{\IEEEauthorrefmark{4} Simons Institute for the Theory of Computing, University of California, Berkeley, Berkeley, CA,  94720 }
 \IEEEauthorblockA{\IEEEauthorrefmark{3} Institut de Physique
   Th\'eorique, CNRS, CEA, Universit\'e Paris-Saclay, F-91191,  Gif-sur-Yvette, France.}
}
\maketitle
\IEEEpeerreviewmaketitle
\begin{abstract}
  We consider the estimation of a $n$-dimensional vector $\vct{x}$
  from the knowledge of noisy and possibility non-linear element-wise
  measurements of $\vct{x} \vct{x}^T$, a very generic problem that
  contains, e.g. stochastic $2$-block model, submatrix localization or
  the spike perturbation of random matrices. We use an interpolation
  method proposed by Guerra \cite{guerra2003broken} and later refined
  by Korada and Macris \cite{korada2009exact}. We prove that the Bethe
  mutual information (related to the Bethe free energy and conjectured
  to be exact by Lesieur {et al.} \cite{lesieur2015mmse} on the basis
  of the non-rigorous cavity method) always yields an upper bound to
  the exact mutual information. We also provide a lower bound using a
  similar technique.  For concreteness, we illustrate our findings on
  the sparse PCA problem, and observe that (a) our bounds match for a
  large region of parameters and (b) that it exists a phase transition
  in a region where the spectum remains uninformative. While we
  present only the case of rank-one symmetric matrix estimation, our
  proof technique is readily extendable to low-rank symmetric matrix
  or low-rank symmetric tensor estimation.
\end{abstract}
\section{Introduction and Main results}
\label{sec:results}
The estimation of low-rank matrices from their noisy, incomplete, or
non linear measurements is a problem that has a wide range of
applications of practical interest in machine learning and statistics,
ranging from the sparse PCA \cite{6875223}, and community detection
\cite{deshpande2015asymptotic} to sub-matrix
localization\cite{chenxu14,hajek2015submatrix}. We shall consider the
setting where the rank-one matrix to be estimated is created as:
\begin{equation}
    \mtx{W}=\frac{1}{\sqrt{n}} \vct{x} \vct{x}^T \, ,\label{modelxx}
\end{equation}
where $\vct{x}$ is a $n$-dimensional vector whose elements were chosen
independently at random from a prior distribution $p({x})$. The matrix
$\mtx{W}$ is then observed element-wise through a noisy non-linear
output channel $P_{\rm out}(Y_{ij}|W_{ij})$, with
$i,j\!=\!1,\dots,n$. We assume the noise to be symmetric so that $Y_{ij}\!=\!Y_{ji}$. 
The goal
is to estimate the unknown vector $\vct{x}$ from $\mtx{Y}$ up to a global
flip of sign, or equivalently 
the unknown rank-one matrix $\mtx{W}$ from $\mtx{Y}$. 
Throughout this paper, we assume that $p({x})$ and $P_{\rm out}(y,w)$ are
independent of $n$, and are known. 

We consider here an information-theoretic viewpoint and analyze the
mutual information for the above model defined as
$I(\mtx{W};\mtx{Y})=I(\vct{x};\mtx{Y})=\EE_{\vct{x}, \mtx{Y}}\log\{
P(\vct{x},\mtx{Y})/[p(\vct{x})P(\mtx{Y})] \}$. Up to a simple term, see
eq.~(\ref{def:fnrg}), the
mutual information is related to the free energy, which is the
fundamental quantity usually considered in statistical physics
\cite{NishimoriBook01,MezardMontanari09}. Recently, an
explicit single-letter characterization of the mutual information
between the noisy observation and the vector to be recovered has been
computed in some special cases of our setting
\cite{6875223,deshpande2015asymptotic}. The general formula has been
derived by Lesieur {\it et al.}
\cite{lesieur2015phase,lesieur2015mmse} on the basis of the heuristic
cavity method from statistical mechanics
\cite{MezardParisi87b,NishimoriBook01,MezardMontanari09}. We shall
refer to the formula conjectured in \cite{lesieur2015mmse} as the {\it
  Bethe mutual information}.  In this contribution, we use a rigorous
technique that also originated in physics, the so-called Guerra
interpolation \cite{guerra2003broken,korada2009exact}, to prove that
the Bethe mutual information provides always an upper bound. By a
variant of the Guerra interpolation, we also provide a lower bound on
the mutual information that matches the upper bound for a sizable range
of parameters. 

\subsection{Main results}
\label{sec:theorems}
Our first result is a rigorous proof of a conjecture from Lesieur {\it et al.}
\cite{lesieur2015mmse}  of channel universality. 
In the context of community detection in graphs with growing average degrees, 
an equivalence between Bernoulli channel and Gaussian channel has been 
proven already in \cite{deshpande2015asymptotic}. 

\begin{thm}[Channel Universality]   \label{th:universlity}
  Assume model (\ref{modelxx}) with a prior
  $p(x)$ having a finite support, and the
  output channel $P_{\rm out}(y|w)$ such that at $w\!=\!0$, $\log P_{\rm out}(y|w)$
  is thrice differentiable with bounded second and third derivatives and $\mathbb{E}_{P_{\rm out}(y|0)} [  | \partial_w \log  P_{\rm out}(y|w)|_{w=0} |^3] =O(1)$. 
  Then the mutual information per variable satisfies 
  \begin{equation}
  I(\mtx{W};\mtx{Y} ) = I (\mtx{W}; \mtx{W}+ \sqrt{\Delta} \; \mtx{\xi} )  +O(\sqrt{n} ), 
  \end{equation}
  where $\mtx{\xi}$ is a symmetric matrix such that $\xi_{ij} \iiddistr \calN(0,1)$ for $i \le j$, 
  and $\Delta$  is the inverse Fisher information (evaluated at $w\!=\!0$) of the
  channel $P_{\rm out}(y|w)$:
  \begin{equation}
    \frac{1}{\Delta} \equiv \mathbb{E}_{P_{\rm out}(y|0)} \left[      \left( \frac{\partial \log P_{\rm out}(y|w)}{\partial w}\Big|_{y,0}\right)^2\right]\, . \label{FisherInformation}
  \end{equation}
\end{thm}
Informally, this means that we only have to compute the mutual information for a  Gaussian channel to
take care of a wide range of channels. 

Our next result is that, the Bethe mutual
information is always an upper bound to the true one for any finite $n$:
\begin{thm}[Upper Bound] \label{th:Upperbound}
  Assume model (\ref{modelxx}) with a prior $p(x)$ having finite support, 
  and  a Gaussian channel 
  such that $P_{\rm out} (y|w)$ is the probability density function of
  a centered Gaussian distribution with variance $\Delta$. 
  Then for all non-negative parameter $m$, the 
  mutual information per variable $I(\vct{x};\mtx{Y})/n$ is {\it upper
  bounded} 
  by the Bethe mutual information $i_{\rm B}(m)$ defined by
\be i_{\rm B}(m)\!=\!\frac {m^2 \!+\!  \[[\EE_{x}(x^2)\]]^2}{4\Delta}
\!-\EE_{x,z}
\!\[[ {\cal J}\!\!\(( \frac{m}{\Delta} , \frac{m x}{\Delta}
\!+\!\sqrt{\frac{m}{\Delta}} z\)) \!\]],
  \label{eq:irs}
  \ee
where $\EE_x$ denotes the expectation taken over the prior distribution $p(x)$, 
$z$ is a Gaussian variable following ${\cal N}(0,1)$, and where
  \be
  {\cal J}(A,B) = \log \int e^{Bx - Ax^2/2} p (x) \diff x
  \, \label{eq:J} .
  \ee 
\end{thm}
Note that the mutual information is related to the free energy only by
a simple term eq.~(\ref{def:fnrg}).

Our last main result yields an asymptotic lower bound:
\begin{thm}[Lower Bound]   \label{th:Lowerbound}
  With the same hypothesis as in Theorem \ref{th:Upperbound}, denote
  $\hat m$ the minimizer of eq.~(\ref{eq:irs}). 
  Define
    \be
  i_{\rm L}(m) \!=\! \frac{2m^2\!-\!{\hat m}^2\!+\! \[[ \EE_{x} (x^2)\]]^2}{4\Delta} -
\EE_{x,z}
\!\[[ {\cal  J}\!\!\(( \frac{\hat m}{\Delta} , \frac{m x}{\Delta}
\!+\!\sqrt{\frac{\hat m}{\Delta}} z\)) \!\]]\!,
  \label{eq:iL}
  \ee
  where $z \sim {\cal N}(0,1)$. 
Assume that $i'_{\rm L} (m)=0$ has a finite number of solutions. 
  Then 
  $$
  \liminf_{n\to \infty} \frac{1}{n} I (\vct{x}; \mtx{Y} ) \ge \min_m i_{\rm L} (m).
  $$
  %

\end{thm}
One can verify that $\hat m$ is always a stationary point of
$i_{\rm L}(m)$.  If additionally $\tilde m\!=\!\hat m$, where
$\tilde m\!\equiv\!\text{\rm argmin}~ i_{\rm L}(m)$, then
$ \min_m i_{\rm L} (m)\!=\! i_{\rm B} (\hat m)$ and the Bethe mutual
information asymptotically equals the true one. As we shall see this
is the case for some range of parameters, but not always.

\subsection{Relation to previous works}
For two particular cases of model (\ref{modelxx}),  the mutual
information has previously been proven rigorously \cite{6875223,deshpande2015asymptotic} using the
approximate message passing algorithm and its state evolution \cite{rangan2012iterative}. Remarkably, these were
constructive proofs, with an explicit algorithm that achieves the
minimum mean squared error (MMSE).  The proof technique of
\cite{6875223,deshpande2015asymptotic} does not extend
straightforwardly when the state evolution had more than one fixed
point. Our approach applies to more general class of problems (even when several
fixed points exist), but is not constructive and our lower bound is not always tight.

We rely on two essential contributions. First, we use the
cavity computations of Lesieur {\it et al.} who solved the
problem using statistical physics methods~\cite{lesieur2015mmse}. Our
results are a considerable step towards confirming the full validity of
this approach. We show that the Bethe mutual information always yields
an upper bound, and by our lower bound we confirm that the Bethe
mutual information is exact for a large range of parameters. 
Secondly, our approach is
inspired by the scheme introduced by Korada and
Macris \cite{korada2009exact} for studying the Sherrington-Kirkpatrick
model of spin glasses on the so-called Nishimori line with the Guerra
interpolation \cite{guerra2003broken}. It also crucially exploits the
Nishimori identities \cite{NishimoriBook01} for optimal Bayesian estimation
\cite{iba1999nishimori,korada2010tight,zdeborova2015statistical}. It
is worth to remark that
for problems of Bayes-optimal estimation (on the Nishimori line) the simplest version of the Guerra interpolation provides
an upper bound on the free energy/mutual information while for
standard statistical mechanics models, or for optimization problems, it provides
instead a lower bound \cite{guerra2003broken,franz2003replica}.

While we present only the rank-one version on model (\ref{modelxx}),
our proof is readily extendable to any finite rank, or even to tensor
factorization. Future directions include the extension to
non-symmetric matrices, which are less straightforward. We believe
that our results, once more, give strong credibility to the use of the
replica and the cavity methods for statistical estimation problems.
\section{Application to sparse Rademacher variables}
\begin{figure}[t]
\includegraphics[width=0.5\textwidth]{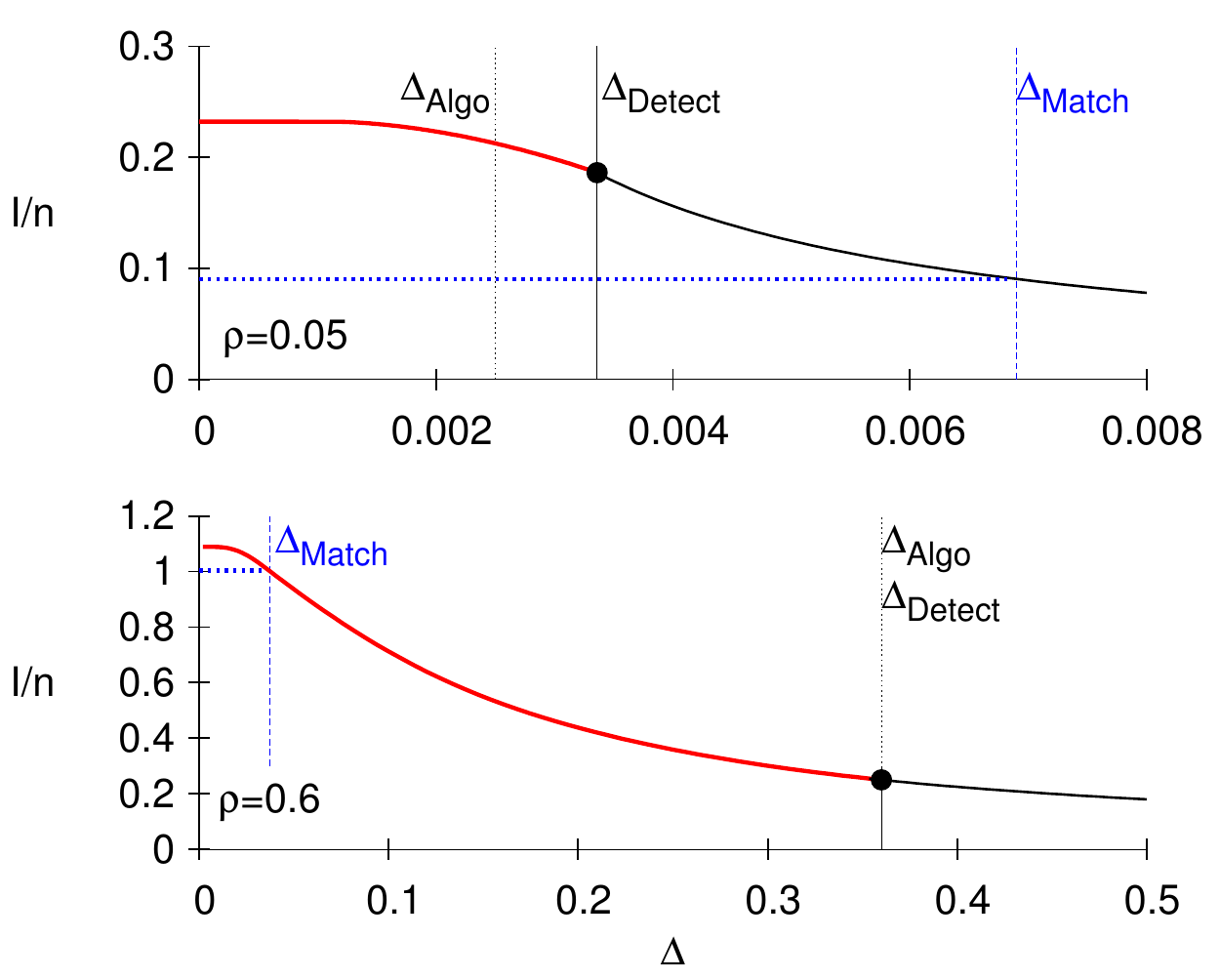}\vspace{-0.2cm}
\caption{Asymptotic mutual information $I/n$ per variable as a
  function of the effective channel noise $\Delta$ for sparse
  Rademacher variables with density 
  $\rho\!=\!0.05$ (top) and $\rho\!=\!0.6$ (bottom). Both bounds are
  tight and equal to the Bethe mutual information for
  $\Delta>\Delta_{\rm Match}(\rho)$. For $\Delta<\Delta_{\rm Match}$,
  however, the lower bound (shown in blue dashed line) is not tight
  and does not coincide with the (conjectured exact) upper bound (full
  line). The detectability transition in the upper bound arises at
  $\Delta_{\rm Detect}$, and is a (conjectured tight) lower bound on
  the true detectability transition. For large enough $\rho$, the
  bounds are tight at the phase transition and we further observe that
  $\Delta_{\rm Detect}=\Delta_{\rm Algo}$ (for instance here when
  $\rho=0.6$).  
  However, we find that when $\rho \approx 0.09$,
  $\Delta_{\rm Detect}$ becomes strictly larger than
  $\Delta_{\rm Algo}$, thus revealing the existence of a regime where
  detection is informationally possible but computationally hard for known polynomial-time algorithms.}\label{figure}
\end{figure}
We shall illustrate (see Fig.~\ref{figure}), for concreteness, our findings on a specific
example of sparse Rademacher variables where
$x\!=\!0,\! +1, -1$ with probability $1\!-\!\rho,\rho/2,\rho/2$
respectively.

Solving numerically for (\ref{eq:irs}) and (\ref{eq:iL}) shows that
the region where the two bounds coincide, and the Bethe mutual
information is thus rigorously exact, is quite sizable. This happens,
for instance, for all $\rho$ if $\Delta$ is larger than a value around
$0.15$, and for all $\Delta$ if $\rho$ is larger than a value around
$0.66$. For large enough $\Delta$, we also find that
$\hat m\!=\!\tilde m\!=\!0$, and a derivative of the mutual information indicates that
it is impossible to find an assignment correlated to the truth.

In the non-sparse case ($\rho\!=\!1$), one can further show that the
minimizers of both (\ref{eq:irs}) and (\ref{eq:iL}) always coincide
and the Bethe mutual information is thus asymptotically exact for all
$\Delta\!>\!0$. A phase transition arises at $\Delta=1$ so that for
$\Delta\!<\!1$, a non trivial solution with
$\hat m\!=\!\tilde m\!>\!0$ appears. Among the problems that belong
into this category are the dense version of the binary stochastic block model
\cite{Massoulie13,Mossel13,deshpande2015asymptotic}, the dense version of the censored block model
\cite{abbe2014decoding,HajekWuXuSDP15,saade2015spectral}, or the
Sherrington-Kirkpatrick model on the Nishimori line as originally
studied by Korada and Macris \cite{korada2009exact}.

When both $\rho$ and $\Delta$ are small enough, numerically we find
that our upper and lower bounds stop to coincide. Define
$\Delta_{\rm match}(\rho)$ as the minimum $\Delta$ so that our upper
and lower bounds match, for a fixed $\rho$. This is illustrated in
Fig.~\ref{figure} with three different values of $\rho$.
When $\Delta\!<\!\Delta_{\rm Algo}\!=\!\sqrt{\rho}$, polynomial-time
algorithms such as message passing
\cite{rangan2012iterative,deshpande2015asymptotic,lesieur2015phase,lesieur2015mmse}
or spectral methods \cite{baik2005phase} are known to be able to find
an assignment with a non-trivial correlation to the truth; thus, in
this region, the non-trivial detection is easy.  We define
$\Delta_{\rm Detect}$ to be the (conjectured) information-theoretic
(IT) threshold for the non-trivial detection.  Depending on the
particular values of $\rho$, we have the following two strikingly different
observations:



(1) When $\rho$ is large (e.g. $\rho\!=\!0.6$),
$\Delta_{\rm Match} \!<\! \Delta_{\rm Algo}\!=\!\Delta_{\rm Detect}$.
In this case, our lower and upper bound coincide, showing that
there is a non-analyticity (phase transition) at
$\Delta_{\rm Detect}$. In this region, the conjectured IT threshold is
indeed the true one and coincides with the algorithmic one
$\Delta_{\rm Algo}$.

(2) When $\rho$ becomes smaller than a certain threshold $\rho^\ast$,
$\Delta_{\rm Algo}\!<\!\Delta_{\rm Detect}\!<\! \Delta_{\rm Match}$.
Numerically we find that $\rho^\ast \approx 0.09$.  The derivative of
$i_{\rm B} $ undergoes a phase transition 
at $\Delta_{\rm Detect}$.  It readily implies that the derivative of
the true mutual information per variable must exhibit a phase
transition somewhere between $\Delta_{\rm Match}$ and
$\Delta_{\rm Detect}$, which is strictly above $\Delta_{\rm
  Algo}$.
Hence, in this case, there exists a region where the non-trivial
detection is informationally possible, but it is impossible via
standard polynomial-time algorithms like spectral methods or message passing.

%
\section{Channel universality}
Let us now show that in order to characterize the mutual information
per variable, it suffices to consider an equivalent Gaussian
channel. We give a detailed rigorous proof in Appendix
\ref{sec:app_channel} and present here only its main idea. We assume
that the prior $p(x)$ has a finite support and denote 
\bea
S_{ij} \equiv \partial_w \log  P_{\rm out}(Y_{ij}|w)|_{w=0} \, , \\
S_{ij}' \equiv \partial^2_w  \log P_{\rm out}(Y_{ij}|w)|_{w=0}\, .
\eea
We assume $ \EE_{Y_{ij}|0} [|S_{ij}|^3 ],$ $S'_{ij},$  and $\partial^3_w  \log P_{\rm out}(Y_{ij}|w)|_{w=0}$ are all bounded.

Note that $W_{ij}=O(1/\sqrt{n})$. 
Using Taylor's expansion of $\log   P_{\rm out}(y |w)|$ at $w=0$, 
for all $i\le j$ we can write
\be
P_{\rm out}(Y_{ij} |W_{ij} ) = P_{\rm out}(Y_{ij} |0) e^{ W_{ij} S_{ij}  +   \frac{1}{2} W^2_{ij}S_{ij}' + O(n^{-3/2}) }.
  \nonumber
\ee
Thus, 
\bea
P_{\rm out}( \mtx{Y} | \mtx{W} ) =  P_{\rm out}( \mtx{Y} |0) 
e^{ \sum_{i \le j}  ( W_{ij} S_{ij}    + \frac{1}{2}   W^2_{ij}S_{ij}' ) + O(\sqrt{n} ) } ,
 \label{eq:equiv}
\eea

Classical properties of the Fisher information give that 
 $\EE_{Y_{ij} | 0 } [S'_{ij} ] =  -\EE_{Y_{ij} | 0 } [S_{ij}^2] = -1/\Delta$. 
Using the fact that
$P_{\rm out}(y|w)$ is close to $P_{\rm out}(y|0)$, one can further
argue that $
\EE_{Y_{ij} |W_{ij} } [ S'_{ij}]  =-1/\Delta+ O(n^{-1/2}) .$ 
By concentration inequalities, we expect that 
$$ 
 \sum_{i \le j} W^2_{ij}S_{ij}' \approx  \sum_{i \le j} W_{ij}^2 \EE_{Y_{ij} |W_{ij} } [ S'_{ij}] = \sum_{ i \le j} W_{ij}^2/ \Delta
+O(\sqrt{n}).
$$
Thus
$$
P_{\rm out}(\mtx{Y} | \mtx{W} ) \approx  P_{\rm out}(\mtx{Y} |0) 
e^{ \sum_{i \le j}  ( W_{ij} S_{ij}    - \frac{1}{2\Delta}   W^2_{ij})  + O(\sqrt{n} )}  
$$
and consequently 
$$
P( \mtx{W} | \mtx{Y} ) \propto P(\mtx{W} ) e^{-\frac{1}{2\Delta} \sum_{i \le j} ( \Delta S_{ij} - W_{ij})^2  +O(\sqrt{n})}.
$$

%
Hence, we expect that
$$
I( \mtx{W} ;  \mtx{Y} ) = I( \mtx{W} ;  \mtx{W} + \sqrt{\Delta}  \; \mtx{\xi} )+ O(\sqrt{n}). 
$$
In other words,  the mutual information per variable $I(\vct{x};\mtx{Y} )/n$
is asymptotically equal to the mutual information per variable of a 
Gaussian channel with noise variance $\Delta$. 

\section{Proving the upper bound}
\label{sec:upperbound}
\subsection{Mutual information and free energy}
Using the channel universality, we only need to deal  with the Gaussian
output. The mutual information between the observation $\mtx{Y}$ and the
unknown vector $\vct{x} $ is defined using the entropy as $I(\vct{x};\mtx{Y}) = H(\mtx{Y})-H(\mtx{Y}|\vct{x})$.
For the Gaussian noise, a
straightforward computation shows that the mutual information per
variable is expressed as
\be
\frac{I(\vct{x};\mtx{Y} )}n = f + \frac{[\EE_{x} (x^2)]^2}{4\Delta}\, ,
\label{def:fnrg}
\ee
where $f=-\EE_{\mtx{Y}} \left[ \log Z(\mtx{Y}) \right]/n$ is the average free energy
per variable using the standard statistical physics terms, and
$Z(\mtx{Y})$ is the partition function defined by
\be Z(\mtx{Y}) \equiv \int {\rm d} \mtx{x} \,  p(\mtx{x}) \exp \left[ \sum_{i \le j} \left( - \frac{x^2_i
    x^2_j}{2n\Delta} + \frac{x_i x_j Y_{ij}} {\sqrt{n} \Delta} \right)
\right]\, . \ee
We now show how to upper bound the free energy $f$.
%

\subsection{Denoising}
\label{sec:denoising}
We first solve a simpler denoising problem. Assume we
observe a noisy version of a vector $\vct{x^\ast}$ that we denote $\vct{y}$:
\bea
\vct{y} &=& \vct{x^\ast} + \sigma \vct{z}, 
\label{eq:denoising}
\eea
where $x^\ast_i \iiddistr p(x)$ and $z_i \iiddistr {\cal N}(0,1)$.  
The corresponding posterior distribution
reads
\be
P(\vct{x} |\vct{y} )  = \frac{1}{Z_0} p( \vct{x} ) \exp \left( -\frac{\| \vct{x} \|_2^2}{2 \sigma^2 } + \frac{ \Iprod{ \vct{x}}{ \vct{y} } }{\sigma^2}  \right)\, .
\ee
For future convenience we denote the variance
$\sigma^2 \equiv \frac \Delta{m}$, where $m$ and $\Delta$ are some
so-far unspecified parameters.
For this denoising problem, the averaged free energy per variables reads
\be
-nf_0 =\EE[ \log Z_0] = \EE_{ x^\ast, z} \[[ {\cal J} \(( \frac{ m}{\Delta}
,   \frac{ m x^\ast}{\Delta} +\sqrt{\frac{ m}{\Delta}} z\)) \]]
\label{eq:frng_den},
\ee
where $x^\ast \sim p(x)$, $z \sim \calN(0,1)$, and ${\cal J}(A,B)$ is the function defined in
eq.~(\ref{eq:J}). Notice how this yields a formula very close to the
one in Theorem~\ref{th:Lowerbound}.

\subsection{The interpolation method}
\label{sec:interpolation}
We now use the Guerra interpolation method, setting an artificial
parameter $t$, where we interpolate between the denoising problem at
$t=0$ and the desired matrix factorization one at $t=1$. To do so, 
assume that we have access to two types of noisy observations: (1) A noisy version
of $\vct{x^\ast}$, as in eq.~(\ref{eq:denoising}), with
now $\sigma^2\!=\!  \frac{\Delta}{m\!(1-t)}$; and (2) a noisy version of
$W_{ij}=x^\ast_i x^\ast_j/\sqrt{n}$ with a Gaussian noise of variance
$\Delta/t$. The posterior distribution, in this case, is given by 
\bea
&&P_t(  \vct{x} | \mtx{Y}  , \mtx{y} ) = \nonumber  \\
&& \frac 1{Z_t} p(\mtx{x}) e^{ t \sum_{i \le j} \left[ -\frac{x^2_i x^2_j}{2 n \Delta} + \frac{x_i x_j
    Y_{ij}}{\sqrt{n}\Delta}  \right] + (1-t)\[[-\frac{m \|\vct{x}\|_2^2}{2\Delta} + \frac{ m \Iprod{ \vct{x}}{ \vct{y} } } {\Delta}\]] }. \nonumber 
\eea
This model interpolates between the denoising problem at
$t=0$ and the one of the matrix factorization problem
at $t=1$. Using the fundamental theorem of algebra,
we write
\be
-n f_1 = \EE_{\mtx{x^\ast},\vct{z}, \mtx{\xi }}\[[ \log{Z_1} \]] =
f_0 - \int_0^1 dt \frac{d}{dt} \EE_{\mtx{x^\ast},\vct{z}, \mtx{\xi }}\[[
\log{Z_t} \]].
\label{eq:guerra}
\ee
The free energy at $t\!=\!0$ is precisely given by
(\ref{eq:frng_den}). Using now eq.~(\ref{def:fnrg}) and
eq.~(\ref{eq:irs}) we write
\be
\frac{I( \vct{x}; \mtx{Y} )}n\!=\! i_{\rm B}(m) - \frac {m^2}{4\Delta} -
\frac 1n\!\int_0^1\! dt \frac d{dt} \EE_{\mtx{x^\ast},\vct{z}, \mtx{\xi }} \[[\log  Z_t \]]\!\! .
\ee
%
Theorem \ref{th:Upperbound} follows from the following lemma:

\begin{lemma} \label{lemma:pos}
For all positive $n$ and $ t \in [0,1]$, we have
\be
 \frac 1n \frac d{dt} \EE_{\mtx{x^\ast}, \vct{z},  \mtx{\xi }} \[[\log  Z_t \]] \ge  - \frac {m^2}{4\Delta} .
\ee
\end{lemma}
%

\subsection{The proof}
\label{sec:proof}

Define 
\begin{align*}
& \calH_t (\vct{x}, \vct{x^\ast}, \mtx{\xi}, \vct{z})  = \sum_{i \le j} -\frac {t x_i^2 x_j^2}{2\Delta n} +
\frac{t x_i x_j x_i^\ast x_j^\ast }{\Delta n} + x_i x_j  \sqrt{ \frac{t}{n \Delta}}  \xi_{ij} 
    \\
    &~  - \frac{ (1-t) m \| \vct{x} \|_2^2}{2\Delta} + 
\frac{(1-t) m \Iprod{\vct{x}}{\vct{x^\ast}}  }{\Delta} + \Iprod{\vct{x}}{\vct{z}} \sqrt{\frac{m (1-t)}{\Delta }} .
\end{align*}
Then $P_t(\vct{x} | \vct{x^\ast}, \vct{z},\mtx{\xi} )  = p(\vct{x})e^{\calH_t}/Z_t$. 
Now we need to compute $\frac{d}{dt} \EE_{\vct{x^\ast}, \vct{z}, \mtx{\xi}} \left[   \log Z_t \right]$.
Notice that 
$$
 \partial_t \calH_t  =  \sum_{i\le j} {\cal A}_{ij} +\sum_i {\cal B}_{i}.
 $$
where
\begin{align}
{\cal A}_{ij} &=&  -\frac {x_i^2 x_j^2}{2\Delta n} +
\frac{x_i x_j x_i^\ast x_j^\ast}{\Delta n} + \frac {x_i x_j}{2\sqrt{n \Delta
    t}} \xi_{ij}\, ,  \label{eqA1}\\
{\cal B}_{i}  &= &m \frac{x_i^2}{2\Delta} - m
\frac{x_ix_i^\ast}{\Delta} - \frac{x_iz_i }{2}  \sqrt{\frac{m}{\Delta (1-t)}} . 
\label{eqB1}
\end{align}
Since the prior $p(x)$ has a finite support $ |  \calA_{ij} |  $ and $ | \calB_i |$ are
dominated by functions integrable 
with respect to $P_t$.  
Thus by the dominated convergence theorem, 
$$
\frac{d}{dt} \log Z_t = \EE_t [ \partial_t \calH_t ] 
=  \sum_{i\le j}  \EE_t [ {\cal A}_{ij}]  +\sum_i  \EE_t [ {\cal B}_{i} ]. 
$$ 
Moreover, $ \EE_t [   \calA_{ij} ]  $ and $ \EE_t [\calB_i ]$ are
dominated by functions integrable 
with respect to the distribution of $\{ \vct{x^\ast}, \mtx{\xi}, \vct{z} \}$ .  
Thus again by the dominated convergence theorem, 
%
\begin{align}
&\frac{d}{dt}\EE_{ \vct{x^\ast}, \vct{z },  \mtx{\xi } } \[[\log Z_t\]]   =  \EE_{\vct{x^\ast}, \vct{z }, \mtx{\xi }} \left[ \frac{d}{dt} \log Z_t \right] \nonumber \\
& = \EE_{ \vct{x^\ast}, \vct{z }, \mtx{\xi }} \[[ \sum_{i\le j}  \EE_t [ {\cal A}_{ij}]  +\sum_i  \EE_t [ {\cal B}_{i} ] \]] \label{eqderiv}. 
\end{align}
We then compute $ \EE_{\xi_{ij}}  \left[  \EE_t[ x_i x_j  ] \xi_{ij}  \right] $ and 
$\EE_{z_i} [ \EE_{t}[x_i ] z_i]$. We use the integration by part to get rid of
the $z_i$ and $\xi_{ij}$. In particular, for a standard Gaussian random variable $z$ and 
a  continuous differentiable function $f(a)$ such that $f(a) \eexp^{-a^2/2} \to 0$ as $a \to \infty$, 
we have that $\expect{ z f(z) } = \expect{ f'(z) }.$
Notice that $P_t$ is a function of $ \mtx{\xi }$ and $\vct{z}$. 
Also, 
$
\partial_{\xi_{ij}} \calH_t = \sqrt{ \frac{t}{n\Delta} } x_i x_j . 
$
Then $| \partial_{\xi_{ij}} \calH_t | $ is dominated by a function integrable under $P_t$. 
By the dominated convergence theorem, 
$$
\partial_{\xi_{ij}} Z_t =   \sqrt{ \frac{t} {n\Delta} } Z_t \;  \EE_t [x_i x_j ]. 
$$
It follows that 
$$
 \partial_{\xi_{ij}}   \EE_t[ x_i x_j  ]  = \sqrt{ \frac{t}{n\Delta} } \left( \EE_t [ x_i^2 x_j^2  ] - \left( \EE_t [ x_i x_j ] \right)^2 \right).
$$
Thus $ \partial_{\xi_{ij}}   \EE_t[ x_i x_j  ] $ is continuous in $\xi_{ij}$. 
Applying the integration by parts, it yields that 
\begin{align*}
 \EE_{\xi_{ij}}  \left[  \EE_t[ x_i x_j  ] \xi_{ij}  \right]  & = \EE_{\xi_{ij}} \left[  \partial_{\xi_{ij}}   \EE_t[ x_i x_j  ] \right] \\
 & = \sqrt{ \frac{t}{n\Delta} }  \EE_{\xi_{ij}} \left[   \EE_t [ x_i^2 x_j^2  ] - \left( \EE_t [ x_i x_j ] \right)^2 \right].
\end{align*}
Similarly, one can show that 
\begin{align*}
\EE_{z_i} [ \EE_{t}[x_i ] z_i] & = \EE_{z_i} \left[ \partial_{z_i} \EE_t [ x_i ] \right]  \\ 
&=\sqrt{\frac{m (1-t)}{\Delta }} \EE_{z_i} \left[ \EE_t [ x_i^2 ] - \left( \EE_t [x_i ] \right)^2 \right].
\end{align*}

It follows from \eqref{eqA1} and \eqref{eqB1} that
\bea
&\EE_{t,\mtx{\xi}} [ {\cal A}_{ij} ] = \EE_{t,\mtx{\xi}}  \[[ 
\frac{x_i x_j x_i^\ast x_j^\ast}{\Delta n} - \frac {x_i x_j \EE_t[x_i x_j]}{2n \Delta
    }  \]] ,\label{eqA2}
\\
& \EE_{t,{\bf z}} [ {\cal B}_{i}  ] = \EE_{t,{\bf z}} \[[-
\frac{m}{\Delta} x_ix_i^\ast + \frac{m}{2\Delta} x_i\EE_t[x_i]
\]] .\label{eqB2}
\eea
Using the Nishimori identities given by Lemma \ref{lmm:nishimori}, we have that 
\bea
 \EE_{\mtx{x^\ast} ,{\bf
    z},  \mtx{\xi } } \[[ \EE_t\[[x_i x_j\]]^2
\]] \!=\! \EE_{\mtx{x^\ast} , \vct{z }, \mtx{\xi } }\[[\EE_t\!\[[x_i x_j x_i^ \ast x_j^\ast  \]]\]] ,\label{Nishi1}
\\
\EE_{\mtx{x^\ast} , \vct{z },  \mtx{\xi }  }\[[ (\EE_t\[[x_i\]])^2
\]] \!=\! \EE_{\mtx{x^\ast} , \vct{z }, \mtx{\xi }  }\[[ \EE_t \[[x_i x_i^\ast \]] \]].
\label{Nishi2}
\eea
Combining \eqref{eqderiv}, \eqref{eqA2}--\eqref{eqB2}, and \eqref{Nishi1}--\eqref{Nishi2} yields that
%
\begin{align*}
& \frac 1n\frac{d}{dt} \EE_{ \vct{x^\ast}, \vct{z },  \mtx{\xi } } \[[\log Z_t\]]  \\
& =\frac 1{2\Delta n^2}
\sum_{i\le j} \EE \[[ x_i x_j x_i^\ast
x_j^\ast\]]  -\frac{m}{2\Delta n} \sum_i  \EE \[[
x_i x_i^\ast \]] \\
& \ge  \frac{\EE[m_t^2] }{4\Delta} - \frac{m \EE[ m_t ] }{2\Delta}  \\
& = \frac{1}{4\Delta} \EE[ (m_t - m)^2 ] - \frac{m^2}{4\Delta} \ge -\frac{m^2}{4\Delta},
\end{align*}
where $m_t =  \Iprod{ \vct{x}}{ \vct{x^\ast} }$ with $ \vct{x}$ drawn from $P_t$.  Lemma \ref{lemma:pos} 
readily follows. 
%

\section{Proving the lower bound}
\label{sec:lowerbound}
The proof for the lower bound also relies on the 
interpolation method. 
Again, the proof idea is inspired by
\cite{korada2009exact}.
\subsection{An ad-hoc model}
We shall first compute the free energy of a totally artificial
model, that does not correspond to any Bayesian inference problem. Later, we
will interpolate the desired free energy starting from this model at
$t=0$.  
Let  $\hat m$ denote the minimizer of eq.~(\ref{eq:irs}).
For a fixed set of $(\vct{z},\mtx{x}^\ast)$, %
let
\bea
\tilde Z_0 = \int {\rm d}\mtx{x}\,  p (\mtx{x})
 e^{ \frac{1}{n \Delta}  \sum_{i \le j}   \[[ x_i x_j  x_i^\ast x_j^\ast \]]
 -\frac{  \hat m  \|\vct{x}\|_2^2}{2 \Delta} + \sqrt{\frac {\hat m}{\Delta } }    \Iprod{\vct{x} }{ \vct{z} } } \nonumber \\
= \int {\rm d}\mtx{x}\,  p (\mtx{x}) e^{\frac 1{2n\Delta}  \Iprod{ \vct{x}  }{ \vct{ x^\ast} }^2 -\frac{  \hat m  }{2 \Delta}  \|\vct{x}\|_2^2 + \sqrt{\frac {\hat m}{\Delta } }    \Iprod{\vct{x} }{ \vct{z} }  + O(1) }  \,, \nonumber
\eea 
where $O(1) = \frac{1}{2n\Delta} \| \vct{x} \|_2^2$. 

Using the Gaussian identity
$e^{\frac{b^2}{4a}}=\sqrt{a/\pi} \int\!  e^{-am^2+bm} \diff m$, with $a=n/{2\Delta}$ and
$b= \Iprod{ \vct{x}  }{ \vct{ x^\ast} } /\Delta,$ we reach
\be
\tilde Z_0\!\! \lesssim \!\!\sqrt{\frac{n}{2\pi \Delta}}\int\! \diff \mtx{x}
\, p (\mtx{x})  \int \diff m
e^{-\frac{n m^2}{2\Delta}
+\Iprod{\vct{ x} }{ \frac{m}{\Delta} \vct{ x^\ast}  + \sqrt{\frac {\hat m}{\Delta    }} \vct{z} }-\frac{\hat m}{2\Delta} \|\vct{x}\|_2^2  }, \nonumber
\ee
where $a_n \lesssim b_n$ means $a_n=O(b_n)$.  
We now invert the integral by Fubini's theorem so that
\bea
\tilde Z_0\!\! \lesssim \!\!\sqrt{\frac{n}{2\pi \Delta}}\int\! \diff m \exp\!{\[[\frac{-n
    m^2}{2\Delta}\!+\! \sum_i  {\cal J}\!\!\(( \frac{\hat m}{\Delta} ,
  \frac{ m x_i^\ast}{\Delta}
\!+\!\sqrt{\frac{\hat m}{\Delta}} z_i\)) \]]} .\nonumber
\eea
Then, a naive application of the Laplace method suggests that 
\begin{align}
\tilde f_0 & := -\limsup_{n \to \infty} \frac{1}{n} \EE_{\vct{z},{\mtx{x^\ast}}} \left[ \log
  { \tilde Z_0} \right]\\
& \ge \text{\rm min}_m\[[
\frac{m^2}{2\Delta} -
\EE_{x^\ast,z}
\!\[[ {\cal  J}\!\!\(( \frac{\hat m}{\Delta} , \frac{m x^\ast}{\Delta}
\!+\!\sqrt{\frac{\hat m}{\Delta}} z\)) \!\]]\]]\, .
\label{fnrg:adhoc}
\end{align}
A rigorous proof that $\tilde f_0$ is indeed lower bounded by the above
expression (which is only what we require) is presented in Appendix
\ref{sec:app_Laplace}, under the assumption that $i'_{\rm L} (m)=0$ has 
a finite number of solutions. 
\subsection{Interpolation reloaded}
The proof of the lower bound then proceed again via the interpolation
method, where we interpolate between the ad-hoc model and the matrix
factorization one by considering the following partition function, at
fixed value of $\{{ \vct{z}, \mtx{\xi}, \vct{x^\ast }\}}$:
 \bea \tilde Z_t &=& \int {\rm d}\mtx{x}\,  p (\mtx{x})
e^{\frac{1}{n\Delta}  \sum_{i \le j} 
  \[[- t\frac{x^2_i x^2_j}{2 } + x_i x_j x^i_0 x^j_0+ \sqrt{n\Delta t} x_i x_j 
    \xi_{ij} \]]}\nonumber \\
&&~~~\times e^{(1-t) \[[-\frac{  \hat m  \| \vct{x}\|_2^2}{2 \Delta} + 
  \((\sqrt{\frac {\hat m}{\Delta (1-t)}} \Iprod{\vct{x}}{\vct{z}}
  \)) \]]} 
\label{upp.eq}\, .
\eea 
This is again detailed in Appendix \ref{sec:app_intepolation}.

\section*{Acknowledgments}
We thank J. Barbier, T. Lesieur, N. Macris and C. Moore for helpful
discussions. Part of the research has received funding from the
European Research Council under the European Union’s 7th Framework
Programme (FP/2007-2013/ERC Grant Agreement 307087-SPARCS).
\ifCLASSOPTIONcaptionsoff \fi
\bibliographystyle{IEEEtran}
\bibliography{bibliography}

\begin{thebibliography}{10}
\providecommand{\url}[1]{#1}
\csname url@samestyle\endcsname
\providecommand{\newblock}{\relax}
\providecommand{\bibinfo}[2]{#2}
\providecommand{\BIBentrySTDinterwordspacing}{\spaceskip=0pt\relax}
\providecommand{\BIBentryALTinterwordstretchfactor}{4}
\providecommand{\BIBentryALTinterwordspacing}{\spaceskip=\fontdimen2\font plus
\BIBentryALTinterwordstretchfactor\fontdimen3\font minus
  \fontdimen4\font\relax}
\providecommand{\BIBforeignlanguage}[2]{{%
\expandafter\ifx\csname l@#1\endcsname\relax
\typeout{** WARNING: IEEEtran.bst: No hyphenation pattern has been}%
\typeout{** loaded for the language `#1'. Using the pattern for}%
\typeout{** the default language instead.}%
\else
\language=\csname l@#1\endcsname
\fi
#2}}
\providecommand{\BIBdecl}{\relax}
\BIBdecl

\bibitem{guerra2003broken}
F.~Guerra, ``Broken replica symmetry bounds in the mean field spin glass
  model,'' \emph{Communications in mathematical physics}, vol. 233, p.~1, 2003.

\bibitem{korada2009exact}
S.~B. Korada and N.~Macris, ``Exact solution of the gauge symmetric p-spin
  glass model on a complete graph,'' \emph{Journal of Statistical Physics},
  vol. 136, no.~2, pp. 205--230, 2009.

\bibitem{lesieur2015mmse}
T.~Lesieur, F.~Krzakala, and L.~Zdeborov{\'a}, ``Mmse of probabilistic low-rank
  matrix estimation: Universality with respect to the output channel,''
  \emph{arXiv preprint arXiv:1507.03857}, 2015.

\bibitem{6875223}
Y.~Deshpande and A.~Montanari, ``Information-theoretically optimal sparse
  pca,'' in \emph{Information Theory (ISIT), 2014 IEEE International Symposium
  on}, June 2014, pp. 2197--2201.

\bibitem{deshpande2015asymptotic}
Y.~Deshpande, E.~Abbe, and A.~Montanari, ``Asymptotic mutual information for
  the two-groups stochastic block model,'' \emph{arXiv:1507.08685}.

\bibitem{chenxu14}
Y.~Chen and J.~Xu, in \emph{Proceedings of ICML 2014 (Also arXiv:1402.1267)}.

\bibitem{hajek2015submatrix}
B.~Hajek, Y.~Wu, and J.~Xu, ``Submatrix localization via message passing,''
  \emph{arXiv preprint arXiv:1510.09219}, 2015.

\bibitem{NishimoriBook01}
H.~Nishimori, \emph{Statistical Physics of Spin Glasses and Information
  Processing: An Introduction}.\hskip 1em plus 0.5em minus 0.4em\relax Oxford
  University Press, 2001.

\bibitem{MezardMontanari09}
M.~M\'ezard and A.~Montanari, \emph{Information, Physics, and
  Computation}.\hskip 1em plus 0.5em minus 0.4em\relax Oxford: Oxford Press,
  2009.

\bibitem{lesieur2015phase}
T.~Lesieur, F.~Krzakala, and L.~Zdeborov{\'a}, ``Phase transitions in sparse
  pca,'' in \emph{Information Theory (ISIT), 2015 IEEE International Symposium
  on}.\hskip 1em plus 0.5em minus 0.4em\relax IEEE, 2015, pp. 1635--1639.

\bibitem{MezardParisi87b}
M.~M{\'e}zard, G.~Parisi, and M.~A. Virasoro.\hskip 1em plus 0.5em minus
  0.4em\relax Singapore: World Scientific.

\bibitem{rangan2012iterative}
S.~Rangan and A.~K. Fletcher, ``Iterative estimation of constrained rank-one
  matrices in noise,'' in \emph{Information Theory Proceedings (ISIT), 2012
  IEEE International Symposium on}.\hskip 1em plus 0.5em minus 0.4em\relax
  IEEE, 2012, pp. 1246--1250.

\bibitem{iba1999nishimori}
Y.~Iba, ``The nishimori line and bayesian statistics,'' \emph{Journal of
  Physics A: Mathematical and General}, vol.~32, no.~21, p. 3875, 1999.

\bibitem{korada2010tight}
S.~B. Korada and N.~Macris, ``Tight bounds on the capacity of binary input
  random cdma systems,'' \emph{Information Theory, IEEE Transactions on},
  vol.~56, no.~11, pp. 5590--5613, 2010.

\bibitem{zdeborova2015statistical}
L.~Zdeborov{\'a} and F.~Krzakala, ``Statistical physics of inference:
  Thresholds and algorithms,'' \emph{arXiv preprint arXiv:1511.02476}, 2015.

\bibitem{franz2003replica}
S.~Franz, M.~Leone, and F.~L. Toninelli, ``Replica bounds for diluted
  non-poissonian spin systems,'' \emph{Journal of Physics A: Mathematical and
  General}, vol.~36, no.~43, p. 10967, 2003.

\bibitem{Massoulie13}
L.~Massouli{\'e}, ``Community detection thresholds and the weak ramanujan
  property,'' in \emph{Proceedings of the 46th Annual ACM Symposium on Theory
  of Computing}.\hskip 1em plus 0.5em minus 0.4em\relax ACM, 2014, pp.
  694--703.

\bibitem{Mossel13}
E.~Mossel, J.~Neeman, and A.~Sly, ``A proof of the block model threshold
  conjecture,'' \emph{arxiv:1311.4115}, 2013.

\bibitem{abbe2014decoding}
E.~Abbe, A.~S. Bandeira, A.~Bracher, and A.~Singer, ``Decoding binary node
  labels from censored edge measurements: Phase transition and efficient
  recovery,'' \emph{Network Science and Engineering, IEEE Transactions on},
  vol.~1, no.~1, pp. 10--22, 2014.

\bibitem{HajekWuXuSDP15}
B.~Hajek, Y.~Wu, and J.~Xu, ``Achieving exact cluster recovery threshold via
  semidefinite programming: Extensions,'' arXiv 1502.07738.

\bibitem{saade2015spectral}
A.~Saade, M.~Lelarge, F.~Krzakala, and L.~Zdeborov{\'a}, ``Spectral detection
  in the censored block model,'' in \emph{Information Theory (ISIT), 2015 IEEE
  International Symposium on}.\hskip 1em plus 0.5em minus 0.4em\relax IEEE,
  2015, pp. 1184--1188.

\bibitem{baik2005phase}
J.~Baik, G.~Ben~Arous, and S.~P{\'e}ch{\'e}, ``Phase transition of the largest
  eigenvalue for nonnull complex sample covariance matrices,'' \emph{Annals of
  Probability}, pp. 1643--1697, 2005.

\bibitem{Latala05}
R.~Lata{\l}a, ``Some estimates of norms of random matrices,'' \emph{Proceedings
  of the American Mathematical Society}, vol. 133, no.~5, p. 1273, 2005.

\bibitem{chatterjee2006}
\BIBentryALTinterwordspacing
S.~Chatterjee, ``A generalization of the lindeberg principle,'' \emph{Ann.
  Probab.}, vol.~34, no.~6, pp. 2061--2076, 11 2006. [Online]. Available:
  \url{http://dx.doi.org/10.1214/009117906000000575}
\BIBentrySTDinterwordspacing

\bibitem{KoradaMontanari11}
S.~B. Korada and A.~Montanari, ``Applications of the lindeberg principle in
  communications and statistical learning,'' \emph{IEEE Transactions on
  Information Theory}, vol.~57, no.~4, pp. 2440--2450, April 2011.

\end{thebibliography}

\newpage

\appendix

\subsection{Proof of channel universality}
\label{sec:app_channel}
Here we present the detailed proof of channel universality. 
In this proof, with a slight abuse of notation, we let $\Iprod{\mtx{A}}{\mtx{B}}
:= \sum_{i \le j} A_{ij} B_{ij}$ for two symmetric matrices $\mtx{A}$ and $\mtx{B}$. 
By
definition
$$
I( \mtx{W}; \mtx{Y} ) = \EE_{ \mtx{W}, \mtx{Y} } \left[ \log \frac{P_{\rm out} (\mtx{Y}|\mtx{W}) }{ \int P( \mtx{W'} ) P_{\rm out} (\mtx{Y}|\mtx{W'} ) \diff \mtx{W'} } \right].
$$
Define  $D_{ij} = S'_{ij} + 1/\Delta$ and 
$\calH(\mtx{W}, \mtx{Y}) =   \Iprod{\mtx{W}}{ \mtx{S} }   - \frac{\| \mtx{W}\|_F^2 } {2 \Delta} + \frac{1}{2} \Iprod{ \mtx{D} } { \mtx{W} \circ \mtx{W}}$,
where $\circ$ denotes the element-wise matrix product.  
In view of \eqref{eq:equiv}, 
\begin{align*}
I(\mtx{W}; \mtx{Y})  = \EE_{\mtx{W}, \mtx{Y}} \left[ \log \frac{ e^{\calH (\mtx{W}, \mtx{Y}) }  }{ \int  P( \mtx{W'} ) e^{\calH ( \mtx{W'}, \mtx{Y}) } \diff \mtx{W'}  } \right] + O(\sqrt{n}) .
\end{align*}
In the following, we compute $ \EE_{\mtx{W}, \mtx{Y}} \left[\calH (\mtx{W}, \mtx{Y}) \right] $  and $ \EE_{\mtx{Y}} \left[ \log \left( \int  P( \mtx{W'} ) e^{\calH ( \mtx{W'}, \mtx{Y}) } \diff \mtx{W'}     \right) \right]$ up to additive errors on the order of $O(\sqrt{n}).$ 

\begin{lemma}
$$
 \EE_{\mtx{W}, \mtx{Y}} \left[\calH (\mtx{W}, \mtx{Y}) \right] =  \frac{n \left( \EE[ x^2] \right)^2  }{4\Delta} + O(\sqrt{n}). 
$$
\end{lemma}
\begin{proof}
Notice that 
\begin{align*}
 & \EE_{\mtx{Y}|\mtx{W}} [  \Iprod{\mtx{W}}{ \mtx{S}  } ]  \\
 &= \frac{n(n+1)}{2} W_{12}  \EE_{Y_{12} |W_{12}}  [ S_{12}  ]  \\
& = \frac{n(n+1)}{2} W_{12}  \EE_{Y_{12} | 0}  [  S_{12} ( 1+ S_{12} W_{12} + O(n^{-1})  ] \\
& = \frac{n(n+1)}{2} W_{12}^2    \EE_{Y_{12} | 0} [S_{12}^2 ] + O (\sqrt{n}) \\
& =  \frac{n+1}{2\Delta} x_1^2x_2^2  + O (\sqrt{n}),
\end{align*}
where we used the fact that $\EE_{Y_{12} | 0}  [  S_{12} ] =0$ and $\EE_{Y_{12} | 0}  [  S^2_{12} ] =1/\Delta.$
It follows that 
$$
 \EE_{\mtx{W}, \mtx{Y}} [\Iprod{\mtx{W}}{ \mtx{S}  } ] = \frac{n+1}{2\Delta} \left( \EE[ x^2] \right)^2 +O(\sqrt{n}).
$$ 
Also, 
\begin{align*}
&\EE_{\mtx{Y}|\mtx{W}} \left[  \Iprod{ \mtx{D} } { \mtx{W} \circ \mtx{W}} \right] \\
& = \frac{n(n+1)}{2} W_{12}^2  \EE_{Y_{12} |W_{12}}  [ D_{12} ] \\
& = \frac{n(n+1)}{2}W_{12}^2  \EE_{Y_{12} |0}  [ D_{12} ( 1  + O(|S_{12}| n^{-1/2} ) ] \\
& = O(\sqrt{n}),
\end{align*}
where we used the fact that $\EE_{Y_{12}|0} [D_{12}] =0$, and $S'_{12}$ and $\EE_{Y_{12}|0} [|S_{12}| ]$ are bounded. 
The lemma readily follows by observing that 
$$
\EE \left[ \mtx{W}\|_F^2 \right] = \sum_{i \le j}  \EE[ x_i^2 x_j^2 ] = \frac{n+1}{2\Delta} \left( \EE[ x^2] \right)^2  
$$
\end{proof}

Define
$$
\Phi(\mtx{S} ) =  \log \left( \int P( \mtx{W'} ) e^{  \Iprod{\mtx{W'} }{S}   - \frac{\| \mtx{W'} \|_F^2 } {2 \Delta} } \diff \mtx{W'}   \right). 
$$
\begin{lemma}
$$
 \EE_{\mtx{Y}} \left[ \log \left( \int P( \mtx{W'} ) e^{\calH ( \mtx{W'}, \mtx{Y}) }  \diff \mtx{W'}    \right) \right] =  \EE_{\mtx{Y}} [ \Phi(\mtx{S}) ] + O(\sqrt{n}).
$$
\end{lemma}
\begin{proof}
Notice that 
$$
| \Iprod{ \mtx{D} } { \mtx{W} \circ \mtx{W}} | = \frac{1}{n} \big| \sum_{i \le j} D_{ij} x_i^2 x_j^2 \big|  \le  \frac{1}{n} \| \mtx{D}\| \| \vct{x} \|_4^4
 = O(\| \mtx{D} \|). 
$$
Thus, 
\begin{align*}
& \EE_{\mtx{Y}} \left[ \log \left( \int P( \mtx{W'} ) e^{\calH ( \mtx{W'}, \mtx{Y}) }  \diff \mtx{W'}   \right) \right]  \\
&=\EE_{\mtx{Y}} [ \Phi(S) ]  + O \left( \EE_{\mtx{Y}} [ \| \mtx{D} \| ]  \right) . 
\end{align*}
Recall that $D_{ij} = \frac{ \partial^2 \log P_{\rm out} (Y_{ij} | 0)}{\partial w^2} + \frac{1}{\Delta}.$ Then $\mtx{D}$ is a symmetric matrix 
where $\{ D_{ij}\}_{i \le j}$ are independent and identically distributed. Moreover, 
$$
\EE_{\mtx{Y}} [ D_{ij} ]  = \EE_{ Y_{ij} |0 } [ D_{ij}  (1 + O ( |S_{ij}| /\sqrt{n} ) ] = O(1/\sqrt{n}),
$$
and since $S'_{ij}$ is bounded, $
\EE_{\mtx{Y}} [ D_{ij}^4 ] =O(1).$
By Latala's theorem~\cite{Latala05}, $\EE_{\mtx{Y}} [ \| \mtx{D} - \EE_{\mtx{Y}}[ \mtx{D} ] \|  ] =O(\sqrt{n})$. Since $\| \EE_{\mtx{Y}} [ \mtx{D} ] \| =O(\sqrt{n})$.
By triangle's inequality, $\EE_{\mtx{Y}} [ \| \mtx{D}  \| ] = O(\sqrt{n})$, and the lemma follows. 
\end{proof}
Finally, we apply the generalized Lindeberg principle \cite{chatterjee2006,KoradaMontanari11} to show the following lemma. 
Let 
$$
\mtx{U}  = \frac{\mtx{W}}{\Delta} + \frac{ \mtx{\xi } }{\sqrt{\Delta}}, 
$$ 
where $\xi_{ij} \sim \calN(0,1)$ for $i \le j$ and $\xi_{ji}=\xi_{ij}$. 
\begin{lemma}
$$
\EE_{\mtx{Y}} [ \Phi( \mtx{S}) ] = \EE_{\mtx{U}} [ \Phi (\mtx{U}) ] +O(\sqrt{n}).
$$
\end{lemma}
Define $a_{ij} = \EE [ S_{ij} ]  - \EE [U_{ij} ]$ and 
$b_{ij} = \EE [ S_{ij}^2] - \EE [ U_{ij}^2]$. 
Note that 
\begin{align*}
\EE_{\mtx{Y}|\mtx{W}} [ S_{ij} ]  & = \EE_{Y_{ij} | 0 } [ S_{ij} ( 1 + W_{ij} S_{ij}  + O(n^{-1} )  )]  \\
& = \frac{W_{ij}}{\Delta} + O(n^{-1}), 
\end{align*}
and 
\begin{align*}
\EE_{\mtx{Y}|\mtx{W}} [ S^2_{ij} ]  & = \EE_{Y_{ij} | 0 } [ S^2_{ij} ( 1   + O( |S_{ij}| n^{-1/2} )  )]  \\
& = \frac{1}{\Delta} + O(n^{-1/2}),  
\end{align*}
where we used the fact that $\EE_{Y_{12}|0} [|S_{ij}|^3 ] =O(1)$. 
Thus $a_{ij} = O(n^{-1})$ and $b_{ij} = O(n^{-1/2})$. Also, one can check that
$\EE [ |S_{ij} |^3 ] + \EE [ |U_{ij}|^3] =O(1)$.
Let $[f(\mtx{W)} ]$ denote the expectation of $f(\mtx{W} ) $ with respect to the measure defined 
by $P(\mtx{W} ) e^{\calH ( \mtx{W}, \mtx{Y}) } \diff \mtx{W}/  \int P(\mtx{W'}) e^{\calH ( \mtx{W'}, \mtx{Y}) } \diff \mtx{W'}  $.  It follows that 
\begin{align*}
\frac{\partial \Phi}{\partial S_{ij} }  &=[ W_{ij} ]  = O(n^{-1/2}) \\
\frac{\partial \Phi^2}{\partial S^2_{ij} } &=[ W_{ij}^2 ] - [W_{ij} ]^2 = O(n^{-1}) \\
\frac{\partial \Phi^3}{\partial S^3_{ij} } & = [ W_{ij}^3] - 3 [W_{ij} ] [W_{ij}^2 ] + 2 [W_{ij} ]^3 = O(n^{-3/2}).
\end{align*}
Therefore, by Lindeberg principle  \cite[Theorem 1.1]{chatterjee2006}, 
\begin{align*}
& \big| \EE_{\mtx{Y}} [ \Phi( \mtx{S} ) ] -  \EE_{\mtx{U}} [ \Phi (\mtx{U}) ] \big| \\
& \le O(n^{-1/2})   \sum_{i\le j} a_{ij} + O(n^{-1}) \sum_{i \le j} b_{ij} + O(\sqrt{n}) \\
&=O(\sqrt{n}). 
\end{align*}

In conclusion, we have shown that 
\begin{align*}
 I(\mtx{W}; \mtx{Y}) & =   \frac{ n \left( \EE[ x^2] \right)^2  }{4\Delta}  -  \EE_{\mtx{U}} [ \Phi (\mtx{U}) ]  +  O(\sqrt{n} )  \\
& =  I(\mtx{W}; \mtx{W}+ \sqrt{\Delta} \; \mtx{\xi} ) +  O(\sqrt{n} ). 
\end{align*}

\subsection{Nishimori identities}
A key ingredient in our interpolation proof is the following Nishimori
identities
\cite{NishimoriBook01,iba1999nishimori,zdeborova2015statistical},
which hold for Bayesian inference problems.

\begin{lemma}[Nishimori identities] \label{lmm:nishimori}
Let $x^\ast$ denote a random sample from a prior distribution $p(x)$,
and we observe $y$ randomly generated from $p(y|x^\ast)$.
Let $x$ and $x'$ denote two independent random samples 
from $p(x|y)$ with the posterior distribution $p(x | y) = p(x) p(y|x)/p(y)$. 
Then for all $f$ such that $\EE[ |f(x, x^\ast)| ] <\infty$, 
$$
\EE[ f(x, x') ] = \EE[ f(x, x^\ast)]. 
$$
\end{lemma}
\begin{proof}
By definition, 
\begin{align*}
 \EE [ f(x, x^\ast)] 
 & = \int p(x^\ast) \int p(y|x^\ast)  \int f(x, x^\ast) p(x|y)  \diff x  \diff y \diff x^\ast \\
& \overset{(a)}{=}\int  p(y) \int \int  f(x, x^\ast) p(x|y) p(x^\ast |y) \diff x \diff x^\ast \diff y \\
& = \EE[ f(x, x')],
\end{align*}
where $(a)$ follows from  the 
fact that $p(x^\ast) p(y|x^\ast)=p(x^\ast|y)p(y)$ and Fubini's theorem. 
\end{proof}

\subsection{Laplace method}
\label{sec:app_Laplace}
We present the rigorous proof. 
For fixed $\vct{x^\ast}$ and $\vct{z}$, 
let 
$$
G(m, \vct{x^\ast}, \vct{z} ) = \frac{-
    m^2}{2\Delta}+ \frac{1}{n} \sum_i  {\cal J}  \left( \frac{\hat m}{\Delta} ,
  \frac{ m x_i^\ast}{\Delta} +\sqrt{\frac{\hat m}{\Delta} } z_i  \right).
$$
and
\begin{align*}
\eta(m, x^\ast_i, z_i ) &= \Delta \times  \partial_{m} \calJ 
\left( \frac{\hat m}{\Delta} ,
  \frac{ m x_i^\ast}{\Delta}
+ \sqrt{\frac{\hat m}{\Delta}} z_i \right) \\
& =  [ x_i] x_i^\ast,
\end{align*}
where $[x_i]$ denotes the mean of $x_i$ under the
distribution proportional to $\exp\left( -  \frac{\hat{m} x^2 }{2 \Delta} +x ( \frac{m x_i^\ast}{\Delta} + \sqrt{ \frac{\hat{m}}{\Delta} } z )    \right) p(x) d(x)$.
Since $p(x)$ has a finite support, $\eta$ is bounded. Without loss of generality, 
assume $|\eta| \le C$. 
It follows that $m$ achieving the maximum value of $G(m, \vct{x^\ast}, \vct{z} ) $ must satisfy 
$$
m = \frac{1}{n} \sum_i \eta(m, x_i, z_i). 
$$
Similarly, $m$ achieving the maximum value of $\expect{G(m,x_1,z_1)}$ must satisfy
$$
m = \expect{\eta(m, x_1,z_1)}. 
$$
Let $\calS$ denote the set of the solutions in $[-C, C]$ of the above fixed point 
equation. By assumption, $i'_{\rm L}(m)=0$ has a finite number of solutions and hence 
$|\calS|$ is finite. 
Notice that  
$$
\partial_m \eta(m, x_i, z_i ) 
=\Delta^{-1} (x_i^\ast)^2 \var(x_i) \ge 0.
$$
It follows that $\eta$ is monotone non-decreasing in $m$. 
Let $\delta=n^{-1/4}$. 
Applying \cite[Lemma 9]{korada2009exact}, we get that
\begin{align*}
& \prob{ \sup_{m \in \reals} \bigg|  \frac{1}{n} \sum_{i} \eta(m, x_i, z_i )  - \EE_{x_1, z_1} [ \eta(m, x_1, z_1 )  ] \bigg| \ge \delta   } \\
 & \le e^{-\Omega(n^{1/4})}.
\end{align*}
Let $\calE$ denote the event that the maximum value of $G(m,x_1,z_1)$ must be attained in the 
set $\calG = \cup_{m \in \calS} (m-\delta, m +\delta)$. It follows that $\prob{\calE^c} \le e^{-\Omega(n^{1/4})}.$
Hence, 
\begin{align*}
& \frac{1}{n} \EE[ \log \tilde{Z}_0] \\
 & \le   \EE [ \max_{ m \in [-C, C] } G(m, \vct{x^\ast}, \vct{z} )  ] \\
& \le \EE[ \max_{ m \in [-C, C] } G(m, \vct{x^\ast}, \vct{z} ) | \calE ] \prob{\calE} \\
& ~~+ \EE[ \max_{ m \in [-C, C] } G(m, \vct{x^\ast}, \vct{z} ) |\calE^c  ] \prob{\calE^c} \\
& \le  \EE[ \max_{ m \in \calG } G(m, \vct{x^\ast}, \vct{z} )] +  \left( O(1) + O\left( \expect{|z_1| } \right) \right) \prob{\calE^c}.
\end{align*}
Taking the limit $n \to \infty$ on both hand sides of the above displayed equation, we have that 
$$
\limsup_{n \to \infty} \frac{1}{n} \EE[ \log \tilde{Z}_0]  \le \limsup_{n\to \infty}  \EE[ \max_{ m \in \calG } G(m, \vct{x^\ast}, \vct{z} )]. 
$$
For $m \in \calG$, $| G' (m, \vct{x^\ast}, \vct{z} )| =O(1).$  Thus, 
$$ 
\max_{ m \in \calG } G(m, \vct{x^\ast}, \vct{z} ) = \max_{ m \in \calS } G(m, \vct{x^\ast}, \vct{z} ) + O(\delta).
$$ 
It follows that
\begin{align}
\limsup_{n \to \infty} \frac{1}{n} \EE[ \log \tilde{Z}_0]  \le \limsup_{n\to \infty}  \EE[ \max_{ m \in \calS } G(m, \vct{x^\ast}, \vct{z} )]. 
\label{eq:logpartitionupper}
\end{align}
Notice that 
\begin{align}
 & \EE[ \max_{ m \in \calS } G(m, \vct{x^\ast}, \vct{z} )]  \nonumber \\
 & \le  \EE\left[ \max_{ m \in \calS }  \big| G(m, \vct{x^\ast}, \vct{z} ) - \expect{G(m, \vct{x^\ast}, \vct{z} ) } \big| \right] \nonumber \\
 &~~+ \max_{ m \in \calS } \expect{G(m, \vct{x^\ast}, \vct{z} ) }. \label{eq:Gupper}
\end{align}
Recall that for any fixed $m \in \calS$, 
\begin{align*}
& \big| G(m, \vct{x^\ast}, \vct{z} ) - \expect{G(m, \vct{x^\ast}, \vct{z} ) } \big|  \\
&= \bigg| \frac{1}{n}  \sum_i   {\cal J}  \left( \frac{\hat m}{\Delta} ,
  \frac{ m x_i^\ast}{\Delta} +\sqrt{\frac{\hat m}{\Delta} } z_i  \right) \\
  &~~ - \expect{  {\cal J}  \left( \frac{\hat m}{\Delta} ,
  \frac{ m x_1^\ast}{\Delta} +\sqrt{\frac{\hat m}{\Delta} } z_1  \right) } \bigg| \\
 & :=T(m). 
\end{align*}
Using the fact that $|x_1^\ast|$ is bounded, one can check that $\calJ \left( \frac{\hat m}{\Delta} ,
  \frac{ m x_1^\ast}{\Delta} +\sqrt{\frac{\hat m}{\Delta} } z_1  \right)$ is sub-Gaussian
  with $O(1)$ sub-Gaussian norm. Thus by Chernoff's bound, for any fixed $m \in \calS$ and 
  $t \ge 0$, 
  $$
  \prob{ T(m) > t} \le e^{- \Omega( n t^2 ) }.
  $$
  By a union bound, it follows that with probability at most $|\calS| e^{- \Omega( n t^2 ) }$, 
  $\max_{m \in \calS T(m)} > t$. 
It follows that 
\begin{align*}
& \expect{\max_{m \in \calS} T(m) } \\
&= \int_{0}^\infty \prob{ \max_{m \in \calS} T(m)> t} \diff t \\
& \le \int_0^{\delta} \diff t + \int_{\delta}^\infty  \prob{ \max_{m \in \calS} T(m)> t} \diff t \\
& \le \delta + |\calS| \int_{\delta}^\infty  e^{- \Omega( n t^2 ) } =o(1). 
\end{align*}
In view of \prettyref{eq:Gupper}, we have that 
$$
\EE[ \max_{ m \in \calS } G(m, \vct{x^\ast}, \vct{z} )] =  \max_{ m \in \calS } \expect{G(m, \vct{x^\ast}, \vct{z} ) } +o(1).
$$
Combining the above display with \prettyref{eq:logpartitionupper}, it yields that 
\begin{align*}
& \limsup_{n \to \infty} \frac{1}{n} \EE[ \log \tilde{Z}_0] \\
& \le \max_{ m \in \calS } \expect{G(m, \vct{x^\ast}, \vct{z} ) } \\
&= \max_{ m } \left \{ \frac{-
    m^2}{2\Delta}+  \EE_{x, z} { {\cal J}  \left( \frac{\hat m}{\Delta} ,
  \frac{ m x^\ast}{\Delta} +\sqrt{\frac{\hat m}{\Delta} } z \right)} \right\},
\end{align*}
which completes the proof. 

\balance
\subsection{Interpolating the lower bound}
\label{sec:app_intepolation}
Let us repeat the interpolation strategy of
sec.~\ref{sec:interpolation}. We start with
eq.~(\ref{upp.eq}). Computing the derivative with respect to $t$, we
find that
\begin{align*}
\frac{d}{dt}\EE_{ \vct{x^\ast}, \vct{z },  \mtx{\xi } } \[[\log \tilde Z_t\]] 
= \EE_{ \vct{x^\ast}, \vct{z }, \mtx{\xi }} \[[ \sum_{i\le j}  \EE_t [ \tilde{ {\cal A}}_{ij}]  +\sum_i  \EE_t [ \tilde{{\cal B}}_{i} ] \]],
\end{align*}
where
\begin{align*}
 \tilde{ {\cal A}}_{ij} &=  -\frac {x_i^2 x_j^2}{2\Delta n} + \frac {x_i x_j}{2\sqrt{n \Delta
    t}} \xi_{ij}  \\
\tilde{{\cal B}}_{i}  &= \frac{ \hat m  }{2\Delta} x_i^2 -
\frac 12 \sqrt{\frac{\hat m}{\Delta (1-t)}}   x_iz_i  
\end{align*}
Performing again the integration by part leads to
\begin{align*}
\EE_{t,\mtx{\xi}} [ \tilde{ {\cal A}}_{ij} ] & = - \EE_{t,\mtx{\xi}} \[[ \frac {x_i x_j \EE_t[x_i x_j]}{2n \Delta
}  \]]  \\
 \EE_{t,{\bf z}} [ \tilde{ {\cal B}}_{i}  ] & = \EE_{t,{\bf z}} \[[ x_i\EE_t[x_i ] \frac{\hat m}{2\Delta}\]] \, . 
\end{align*}
Let $x'$ be an independent copy of $x$. Therefore we have
\bea
&&-  \frac{1}{n} \frac{d}{dt}\EE_{\xi,\vct{z},\mtx{x^\ast} }\!\[[ \log \tilde{Z}_t\]] \nonumber \\
\!\!\!\!\!\!&=& \frac{1}{2n^2\Delta}  \sum_{i \le j} 
\EE [x_i x_j x'_i x'_j ] - \frac{\hat
  m}{2\Delta} \sum_i \EE [x_i x'_i ] \nonumber \\
&\ge & \frac{1}{4n^2\Delta} \EE [ \Iprod{\vct{x}} {\vct{x'}}^2 ] -  \frac{\hat
  m}{2\Delta n} \EE[ \Iprod{\vct{x}}{\vct{x'}} ] \nonumber \\
& =& \frac{1}{4\Delta} \EE \left[  \Iprod{\vct{x}}{\vct{x'}}/n  - \hat{m} \right] - \frac{\hat
  m^2}{4\Delta} \ge - \frac{\hat  m^2}{4\Delta}\nonumber \, ,
\eea
which, with (\ref{eq:guerra}) and (\ref{fnrg:adhoc}), leads to Theorem (\ref{th:Lowerbound}).

\end{document}